\RequirePackage{fix-cm}
\documentclass[ims,authoryear]{imsart}
\RequirePackage{amsthm,amsmath,amsfonts,amssymb}
\RequirePackage{mathtools}
\RequirePackage{natbib}
\startlocaldefs
\numberwithin{equation}{section}
\theoremstyle{plain}
\newtheorem{theorem}{Theorem}[section]
\newtheorem{proposition}[theorem]{Proposition}
\newtheorem{lemma}[theorem]{Lemma}
\newtheorem{corollary}[theorem]{Corollary}
\theoremstyle{definition}

\theoremstyle{definition}
\newtheorem{remark}[theorem]{Remark}
\newcommand{\OO}{\mathrm{O}}
\newcommand{\SO}{\mathrm{SO}}
\newcommand{\so}{\mathfrak{so}}
\newcommand{\R}{\mathbb R}
\newcommand{\T}{\mathbb T}
\newcommand{\E}{\mathbb E}
\newcommand{\PP}{\mathbb P}

\newcommand{\tr}{\operatorname{tr}}
\newcommand{\diag}{\operatorname{diag}}
\newcommand{\Law}{\mathcal L}

\newcommand{\dist}{\operatorname{dist}}
\newcommand{\TV}{\mathrm{TV}}
\newcommand{\dd}{\mathrm d}
\newcommand{\ip}[2]{\langle #1,#2\rangle_*}
\newcommand{\norms}[1]{\lVert #1\rVert_*}
\newcommand{\normF}[1]{\lVert #1\rVert_{\mathrm F}}
\newcommand{\Wc}{\mathcal W_2}
\newcommand{\eps}{\varepsilon}
\newcommand{\dF}{d_{\mathrm F}}
\allowdisplaybreaks[1]
\endlocaldefs

\begin{document}
\begin{frontmatter}
\title{Wasserstein mixing of a systematic-scan random rotation sampler}
\runtitle{Wasserstein mixing of a systematic-scan random rotation sampler}
\runauthor{A. Sepehri}
\begin{aug}
\author[A]{\fnms{Amir}~\snm{Sepehri}\ead[label=e1]{sepehri@alumni.stanford.edu}}
\address[A]{Independent Researcher, San Diego, CA \printead[presep={,\ }]{e1}}
\end{aug}
\begin{abstract}
We study the mixing time of a systematic-scan analogue of Kac's walk that
was proposed as a fast surrogate for Haar-distributed orthogonal matrices in
randomized high-dimensional algorithms and was conjectured to approach Haar
measure after only logarithmically many sweeps. We show that this conjectured
speed-up does not occur for convergence of the full matrix law to Haar
measure in Frobenius Wasserstein distance. At fixed normalized accuracy, the
mixing time lies between order $n/\log n$ and order $n$ sweeps; at fixed
absolute Frobenius accuracy, the corresponding bounds are between order $n$
and order $n\log n$. More strongly, below the scale $n/\log n$, the
normalized Wasserstein distance remains asymptotically at its extremal value.
We also show that the output law is singular with respect to Haar measure for
fewer than $n/2$ sweeps. Thus the sampler may provide effective
application-specific randomization without exhibiting the much faster
full-Haar mixing.
\end{abstract}
\end{frontmatter}

\section{Introduction}
\cite{JOR11} proposed a structured random orthogonal transformation for use in randomized algorithms for high-dimensional data analysis, particularly approximate nearest-neighbor search. Their construction alternates random coordinate permutations with
sequential two-dimensional rotations and includes a Fourier-type transform,
yielding a fast-to-apply surrogate for a dense random orthogonal matrix. The practical motivation is
computational: a dense Haar-distributed orthogonal matrix is expensive to
form and apply, whereas their structured transformation is represented as a
product of permutations, sequential two-dimensional rotations, and a Fourier
transform, so that it can be applied to a vector in
$O(n(\log n+m))$ operations when $m$ rotation sweeps are used. This
structured randomization is useful in its own right and underlies the
computational gains sought in their nearest-neighbor construction.

The authors further suggested that the transformation is observed to be close
to Haar measure once the total number of rotation sweeps satisfies
$M_1+M_2=O(\log n)$ \cite[Remark 2]{JOR11} . If valid in a strong full-distribution
sense, such a mixing rate would have had two notable consequences. First, it would yield an
explicit approximately Haar-distributed $n\times n$ orthogonal matrix in
$O(n^2\log n)$ operations by applying the transformation to the standard
basis, substantially below the cubic cost of standard exact Haar samplers.
Second, from the viewpoint of Markov-chain dynamics, the construction can be
viewed as a systematic-scan analogue of Kac's walk, augmented by additional
scrambling through random coordinate permutations and the Fourier transform.
Whereas Kac's walk selects a coordinate plane at random at each elementary
step, one sweep of the present sampler visits the adjacent planes
\[
(1,2),(2,3),\ldots,(n-1,n)
\]
in a prescribed order. An $O(\log n)$-sweep mixing theorem would therefore
provide an unusually strong example in which a systematic schedule of
low-dimensional updates, together with scrambling, achieves a substantially
faster full-distribution mixing rate than the corresponding random-scan
paradigm.

Our results show that these stronger full-distribution conclusions do not
occur in Frobenius Wasserstein distance. At fixed normalized accuracy, the
number of sweeps lies between order $n/\log n$ and order $n$; at fixed
absolute Frobenius accuracy, it lies between order $n$ and order
$n\log n$. In particular, $O(\log n)$ sweeps do not approximate the full Haar
law in normalized Wasserstein distance. This does not diminish the
computational usefulness of the sampler for its intended applications:
effective randomization for a particular algorithm, marginal, or observable
is a different requirement from approximation of the complete matrix law by
the uniform distribution.

Throughout the paper, $\OO(n)$ denotes the group of real $n\times n$ orthogonal matrices, and $\SO(n)$ its determinant-$+1$ component. Haar measure is the uniform probability measure on these compact matrix groups. The argument uses some elementary geometry of orthogonal matrices, but no prior familiarity with Lie groups is assumed; necessary concepts are introduced in Section~\ref{sec:geometry}.

The mixing time upper bound is based on a coupling argument. The coupling method builds on the local-to-global Wasserstein framework developed by \cite{Oliveira09} for Kac's walk. A sweep here is not a sequence of independently selected coordinate pairs: its planes are adjacent and used in a prescribed order. Thus one cannot obtain the result by simply dividing a Kac-walk bound by $n-1$. The sampler-specific calculation is that the derivatives with respect to all $n-1$ sweep angles form an orthonormal tangent frame. A triangular change of the angles preserves their joint uniform law exactly and removes an average fraction $2/n$ of any infinitesimal squared displacement. We give the local-to-global argument explicitly, including its extension to arbitrary initial laws and its treatment of the two determinant components.

The support-entropy approach to lower bounds also has a precedent in \cite[Section 6]{Oliveira09}. Here we calculate the entropy of the full sequential-sweep support, retaining both its continuous angle parameters and its discrete permutations. For completeness, the needed Haar small-ball estimate is proved directly from the conditional distribution of Haar columns and a Gaussian tail inequality; no concentration theorem is needed for the lower bound. Applying the nonasymptotic estimate at radii of order $n^{-1/2}$ in the normalized metric yields the absolute-accuracy lower bound. Section~\ref{subsec:kaccomparison} compares Oliveira's convention and clock with ours; the normalized sweep estimate is not, by itself, a speedup over Kac's walk.

For orientation, the proof of the upper bound has three ingredients. First, changing the $n-1$ angles in one sweep generates $n-1$ mutually orthogonal infinitesimal directions. Second, averaging over the random permutation makes these directions capture exactly a fraction $2/n$ of the squared size of an arbitrary infinitesimal displacement. Third, a triangular shift of the angles removes those projected directions while preserving the joint uniform distribution of all angles. Oliveira's local-to-global theorem then converts this nearby-point contraction into a Wasserstein contraction for arbitrary initial laws. The lower bound uses a different idea: after $m$ sweeps the output depends on only $m(n-1)$ continuous angles, together with finitely many permutations, so its support has much smaller metric complexity than a typical Haar subset when $m$ is small.

\subsection{The sampler and the precise scope of the results}

Let $e_1,\ldots,e_n$ be the standard basis, and put
\[
 E_{ij}=e_i e_j^{\mathsf T}-e_j e_i^{\mathsf T},\qquad
 R_{ij}(\theta)=\exp(\theta E_{ij})\quad(i<j).
\]
The nontrivial block of $R_{ij}(\theta)$ is
$\left(\begin{smallmatrix}\cos\theta&\sin\theta\\-\sin\theta&\cos\theta\end{smallmatrix}\right)$.
For $\theta\in\T^{n-1}$, where $\T=\R/(2\pi\mathbb Z)$, define
\begin{equation}\label{eq:sweep}
 Q(\theta)=R_{n-1,n}(\theta_{n-1})\cdots
 R_{2,3}(\theta_2)R_{1,2}(\theta_1).
\end{equation}
One random sweep is the matrix $G=Q(\theta)P$, where $P$ is a uniform permutation matrix and the $n-1$ angles are independent uniform variables on $\T$, independent of $P$. Different sweeps use independent randomness.

These are ideal continuous random variables. The results do not assert a property of an unspecified finite-state pseudorandom number generator or an unspecified floating-point implementation. Within this idealization, \eqref{eq:sweep} is the sequential rotation rule in \cite[Section 1.1, equations (1)--(2)]{Sepehri19}; the independence of all angle variables is also explicit in \cite[Section 2.4]{JOR10}.

The Fourier factor in that construction is as follows. Its exact formula is included to identify the original sampler; for the proofs it may simply be regarded as a fixed orthogonal matrix, since our bounds are uniform over predetermined orthogonal insertions. For $d=\lfloor n/2\rfloor$, let
\[
 T_{kl}=d^{-1/2}\exp\{-2\pi\mathrm i(k-1)(l-1)/d\},\qquad
 (Zx)_l=x_{2l-1}+\mathrm i x_{2l}.
\]
On $\R^{2d}$ define $F_{2d}=Z^{-1}TZ$; for odd $n$, extend this map by fixing the last coordinate. Since $T$ is unitary, $F_n$ is orthogonal. Its real determinant is $|\det_{\mathbb C}T|^2=1$, so $F_n\in\SO(n)$.

The original output is a product of $M_2$ sweeps, followed by $F_n$, followed by $M_1$ sweeps, with $m=M_1+M_2$. We prove a slightly more general result. Given deterministic matrices $A_0,\ldots,A_m,x\in\OO(n)$, let
\begin{equation}\label{eq:process}
 X_0=A_0x,\qquad X_t=A_t Q(\theta^{(t)})P_tX_{t-1},
 \quad 1\leq t\leq m,
\end{equation}
and write $\nu_{n,m}^{A,x}=\Law(X_m)$. Choosing $x=I$, all $A_t=I$ except $A_{M_2}=F_n$, recovers the sampler above, with the endpoint interpretations $A_0=F_n$ or $A_m=F_n$ allowed. Relabeling independent sweep variables accommodates either convention for writing the products in \cite{Sepehri19}.

All deterministic factors are fixed in advance and independent of the random variables. They may depend on $n$ and $m$. The lower bounds do not cover arbitrary data-dependent choices of these factors, nor an initial condition containing additional continuous randomness.

\subsection{Probability metrics}

Let $H_n$ denote Haar probability measure on $\OO(n)$, and let $H_{n,+}$ and $H_{n,-}$ be its conditional laws on the two determinant components. For a real matrix $A$, the Frobenius norm is the ordinary Euclidean norm of its entries,
\[
 \normF A=\left(\sum_{i,j}A_{ij}^2\right)^{1/2}=\bigl(\tr(A^{\mathsf T}A)\bigr)^{1/2}.
\]
Write $\dF(U,V)=\normF{U-V}$ for the resulting chordal distance between two orthogonal matrices, and define
\begin{equation}\label{eq:rho}
 \rho(U,V)=\frac{\normF{U-V}}{\sqrt n}.
\end{equation}
For $p\geq1$ and probability measures $\mu,\nu$, write
\[
 W_{p,\rho}(\mu,\nu)
 =\inf_{\pi\in\Pi(\mu,\nu)}
 \left(\int\rho(U,V)^p\,\pi(\dd U,\dd V)\right)^{1/p},
\]
where $\Pi(\mu,\nu)$ is the set of couplings. We take total variation to mean
$\|\mu-\nu\|_{\TV}=\sup_B|\mu(B)-\nu(B)|$, with values in $[0,1]$.
For each fixed $n$, Wasserstein convergence for $\rho$ is equivalent to weak convergence because the group is compact. This fixed-dimension fact does not make mixing rates invariant under dimension-dependent rescaling of the metric.

For every $p\geq1$, rescaling the cost in the same set of couplings gives the exact identity
\begin{equation}\label{eq:Wscaling}
 W_{p,\rho}(\mu,\nu)=n^{-1/2}W_{p,\dF}(\mu,\nu).
\end{equation}
Thus normalized accuracy $\eps$ is precisely absolute accuracy $\eps\sqrt n$. The normalization does not change the chain or its contraction factor; it changes the accuracy scale.

The scale $\sqrt n$ is intrinsic to the matrices: $\normF U^2=n$ for every $U\in\OO(n)$, and the diameters for $\dF$ and $\rho$ are $2\sqrt n$ and $2$, respectively. Moreover,
\begin{equation}\label{eq:normalizationmeaning}
 \rho(U,V)^2
 =\frac1n\sum_{j=1}^n\|Ue_j-Ve_j\|_2^2
 =\E_v\|(U-V)v\|_2^2,
\end{equation}
where $v$ is uniform on the unit sphere. The last equality follows from $\E_v[vv^{\mathsf T}]=I/n$ by taking a trace. Consequently normalized $W_2$ is an optimally coupled root-mean-square discrepancy averaged over input directions, with one matrix coupling shared across directions, not worst-case discrepancy over all directions. A fixed nonzero rotation in one plane illustrates the distinction: its $\dF$ distance from the identity is constant, whereas its $\rho$ distance tends to zero. Dividing by $n$ instead of $\sqrt n$ would make even the diameter tend to zero and would trivialize fixed positive accuracy.

\subsection{Permutation shuffles and systematic scans}
\label{subsec:scans}

There is a finite-group analogue of the distinction between Kac's random
choice of coordinate planes and the ordered rotations in
\eqref{eq:sweep}. For random transpositions, total-variation cutoff occurs at
$\tfrac12 n\log n$ \cite{DS81}. In the cyclic-to-random variant, one
position is visited deterministically while the other is chosen uniformly;
its mixing time remains of order $n\log n$ \cite{Mironov02,MPS04}. Thus a
systematic schedule need not change the asymptotic mixing order.

A closer analogue of the adjacent-plane sweep is provided by adjacent
transpositions. Systematic- and random-scan versions of a biased
adjacent-transposition chain are both known to mix in order $n^2$
elementary updates \cite{DR00,BBHM05}. In a different adjacent-transposition
shuffle, a cyclic systematic scan improves the mixing constant by a factor
of two relative to random scan \cite{Lacoin16,NN19}. Hence the permutation
literature provides examples in which systematic scheduling preserves the
mixing order, and others in which it improves the constant, but not the
kind of order-of-magnitude acceleration that an $O(\log n)$-sweep theorem
would suggest here.

The analogy is structural rather than literal. The present sampler may be
viewed as a systematic-scan analogue of Kac's walk augmented by additional
scrambling: each sweep visits the adjacent planes
\[
 (1,2),(2,3),\ldots,(n-1,n)
\]
in a prescribed order, while an independent uniform permutation is applied
between sweeps and the original construction also contains a fixed Fourier
factor. The permutation is essential to the isotropic projection identity
in Lemma~\ref{lem:isotropy}; removing it gives a different chain.
Comparisons with Kac's walk must also use a common clock: $m$ sweeps contain
$m(n-1)$ elementary plane rotations. Our results therefore do not, by
themselves, imply a speedup over Kac's walk at a fixed number of elementary
rotation updates.

\subsection{Kac's walk: clocks, metrics, and recent results}\label{subsec:kaccomparison}

For reference, one step of Kac's walk on $\SO(n)$ chooses an unordered coordinate pair uniformly, chooses an independent uniform angle, and left-multiplies the current matrix by the corresponding plane rotation. Thus one Kac step contains one elementary rotation, whereas one sweep of our sampler contains $n-1$ prescribed adjacent-plane rotations followed by a random permutation.

\cite[Section 1 and Theorems 1--2]{Oliveira09} uses the unnormalized Hilbert--Schmidt distance $\dF$ and the intrinsic distance $D_{\mathrm{HS}}$ induced by $\langle A,B\rangle_{\mathrm{HS}}=\tr(A^{\mathsf T}B)$ on $\SO(n)$. His upper bound is
\begin{equation}\label{eq:oliveiraabsolute}
 \tau^{\mathrm{Kac}}_{D_{\mathrm{HS}},2}(\eps)
 \leq\left\lceil n^2\log\frac{\pi\sqrt n}{\eps}\right\rceil,
 \qquad 0<\eps<\pi\sqrt n,
\end{equation}
where one step is one random plane rotation and $\tau$ is worst-case mixing time. His lower bound is $\tau^{\mathrm{Kac}}_{\dF,1}(\eps_0)\geq c n^2$ for universal $c,\eps_0>0$. These are unnormalized statements.

The intrinsic and chordal metrics differ by dimension-independent factors:
\begin{equation}\label{eq:HSchordcomparison}
 \dF(U,V)\leq D_{\mathrm{HS}}(U,V)
 \leq\frac\pi2\dF(U,V),\qquad U,V\in\SO(n).
\end{equation}
Our intrinsic metric $D$ in Section~\ref{sec:geometry} satisfies $D_{\mathrm{HS}}=\sqrt2\,D$ on a component; that constant convention is separate from division by $\sqrt n$.

Rescaling \eqref{eq:oliveiraabsolute} gives
\begin{equation}\label{eq:oliveiranormalized}
 \tau^{\mathrm{Kac}}_{D_{\mathrm{HS}}/\sqrt n,2}(\eps)
 =\tau^{\mathrm{Kac}}_{D_{\mathrm{HS}},2}(\eps\sqrt n)
 \leq\left\lceil n^2\log\frac\pi\eps\right\rceil,
 \qquad 0<\eps<\pi.
\end{equation}
The same upper bound applies to $W_{2,\rho}$ by \eqref{eq:HSchordcomparison}. An absolute fixed-accuracy lower bound, however, becomes a normalized lower bound at accuracy $\eps_0/\sqrt n$, not at fixed normalized accuracy.

Table~\ref{tab:clocks} matches the metric and the elementary-rotation clock. A sweep contains $n-1$ rotations and a permutation; the extra permutation and deterministic-insertion costs are not included in the table. Kac's target is Haar measure on $\SO(n)$, whereas our uncorrected sampler targets $\OO(n)$; Proposition~\ref{prop:SO} provides the determinant-corrected version for the same target. These comparisons concern sufficient times, not sharp rates or wall-clock costs. In particular, the normalized $O(n)$ sweep estimate is not a logarithmic improvement over Oliveira's result: after matching accuracy scales and counting rotations, the two coupling upper bounds have the same order.

\begin{table}[tb]
\setpkgattr{tablename}{size}{\upshape\scshape}
\caption{Upper bounds at fixed positive Wasserstein accuracy, with normalized accuracy below $\sqrt2$. Kac's column follows from \eqref{eq:oliveiraabsolute}--\eqref{eq:oliveiranormalized}; the sampler columns follow from Theorem~\ref{thm:upper}.}\label{tab:clocks}
\centering
\begin{tabular}{@{}lccc@{}}
\hline
Metric & \shortstack{Kac's walk\\rotations}
 & \shortstack{Sequential sampler\\sweeps}
 & \shortstack{Sequential sampler\\rotations}\\
\hline
$W_{2,\rho}$ & $O(n^2)$ & $O(n)$ & $O(n^2)$\\
$W_{2,\dF}$ & $O(n^2\log n)$ & $O(n\log n)$ & $O(n^2\log n)$\\
\hline
\end{tabular}
\end{table}

The literature on Kac's walk now gives a sharper picture, but several distinct state spaces and notions of mixing should be kept separate. On the sphere $S^{n-1}$, \cite{PS17} proved total-variation mixing in order $n\log n$, and \cite{JM26} recently proved cutoff, for the walk started from a coordinate vector, at $C_{\mathrm{BRW}}n\log n$ with an explicit constant $C_{\mathrm{BRW}}\approx3.8916$. On the full rotation group $\SO(n)$, \cite{PS18} previously bounded the total-variation mixing time between order $n^2$ and order $n^4\log n$. Very recently, \cite{PS26} improved the upper bound to $O(n^2\log n)$, matching the conjectured order up to constants. A different recent result, \cite{PSV26}, studies ``pseudo-mixing'' on $\SO(n)$: the first $k$ columns mix in Wasserstein distance in $O(n(k+\log n)\log n)$ steps at fixed accuracy, and low-degree polynomial tests become unable to distinguish the walk from Haar on a related scale. These developments reinforce a distinction that is central here: full-matrix mixing, mixing of a single vector or a few columns, and agreement for restricted classes of observables can occur on different time scales.

For the direct quantitative comparison in Table~\ref{tab:clocks}, however, Oliveira's 2009 Wasserstein estimate remains the relevant benchmark because it concerns the same full rotation-matrix chain and a closely comparable transportation metric. 

\section{Main results}\label{sec:main}

The four main statements separate the two sides of the problem. The upper bound is dynamical: repeated sweeps contract nearby laws in Wasserstein distance. The lower bounds are geometric: with too few sweeps, the set of possible outputs is too small, in metric-entropy or dimension terms, to resemble Haar measure globally. We state the results first and postpone all geometric notation needed only for the proofs until Section~\ref{sec:geometry}.

Put
\begin{equation}\label{eq:constants}
 r_n=\sqrt{1-\frac2n},\qquad
 b_n=\pi\sqrt{\frac{2\lfloor n/2\rfloor}{n}}.
\end{equation}

\begin{theorem}[Wasserstein upper bound]\label{thm:upper}
Let $n\geq3$ and $m\geq1$. For every deterministic choice in \eqref{eq:process},
\begin{equation}\label{eq:upper}
 W_{2,\rho}(\nu_{n,m}^{A,x},H_n)
 \leq \min\{\sqrt2,\,b_n r_n^m\}
 \leq \min\{\sqrt2,\,\pi e^{-m/n}\}.
\end{equation}
Consequently, $m\geq\lceil n\log(\pi/\eps)\rceil$ is sufficient for accuracy $0<\eps<\sqrt2$. In the unnormalized Frobenius metric,
\begin{equation}\label{eq:unnormalized}
 W_{2,\normF{\cdot}}(\nu_{n,m}^{A,x},H_n)
 \leq\pi\sqrt{2\lfloor n/2\rfloor}\,r_n^m
 \leq\pi\sqrt n\,e^{-m/n}.
\end{equation}
For $n=2$, the output is exactly Haar after one or more sweeps.
\end{theorem}

The notation $W_{2,\normF{\cdot}}$ in \eqref{eq:unnormalized} is the same as $W_{2,\dF}$. Equations \eqref{eq:upper} and \eqref{eq:unnormalized} are the same estimate in different units. At fixed normalized accuracy $\eps$, $n\log(\pi/\eps)$ sweeps suffice before rounding. At fixed absolute accuracy, the sufficient count is
\[
 n\log\frac{\pi\sqrt n}{\eps}
 =n\log\frac\pi\eps+\frac n2\log n.
\]
The additional term comes from reducing an extra factor $\sqrt n$ by exponential contraction. This explains the logarithm in the sufficient-time bound; it does not prove that the actual mixing times in the two metrics differ by a factor $\log n$. At normalized accuracy $\eps_n=n^{-\alpha}$ with fixed $\alpha>0$, the same bound already requires $O(n\log n)$ sweeps.

\begin{theorem}[Explicit lower bound]\label{thm:lower}
Let $n\geq3$, $m\geq1$, $k=\lfloor n/2\rfloor$, and define
\begin{align}
 B_n&=\min\left\{1,\,
 2^{\,n-k}\left(\frac7{16}\right)^{k(2n-k-1)/4}\right\},\label{eq:Bn}\\
 L_{n,m}&=(n!)^m(40m)^{m(n-1)},\qquad
 \eta_{n,m}=\min\{1,L_{n,m}B_n\}.\label{eq:eta}
\end{align}
Then, for every deterministic choice in \eqref{eq:process} and $p\geq1$,
\begin{equation}\label{eq:lower}
 W_{p,\rho}(\nu_{n,m}^{A,x},H_n)
 \geq\frac14(1-\eta_{n,m})^{1/p}.
\end{equation}
The constants are not optimized.
\end{theorem}

\begin{theorem}[Sublinear-depth obstruction]\label{thm:stronglower}
Suppose $m=m_n\geq1$ and $m_n=o(n/\log n)$. Uniformly over the deterministic choices in \eqref{eq:process},
\begin{equation}\label{eq:stronglower}
 W_{1,\rho}(\nu_{n,m_n}^{A,x},H_n)\longrightarrow\sqrt2,
 \qquad
 W_{2,\rho}(\nu_{n,m_n}^{A,x},H_n)\longrightarrow\sqrt2.
\end{equation}
In particular, a number of sweeps of order $\log n$ does not produce full-matrix Haar approximation in either of these metrics.
\end{theorem}
\begin{remark}[Interpretation of the limit $\sqrt2$]
The constant $\sqrt2$ is extremal for Wasserstein distance from Haar
measure in the normalized Frobenius metric. Indeed,
$W_{2,\rho}(\mu,H_n)\leq\sqrt2$ for every probability measure
$\mu$ on $\OO(n)$, while
\[
  W_{2,\rho}(\delta_V,H_n)=\sqrt2
  \qquad\text{for every }V\in\OO(n).
\]
Thus Theorem~\ref{thm:stronglower} says that when
$m=o(n/\log n)$, the output law is asymptotically no closer to Haar,
in whole-matrix normalized Wasserstein distance, than a deterministic
orthogonal matrix. This does not preclude substantial randomization of
particular marginals or observables.
\end{remark}
\begin{theorem}[Singularity]\label{thm:TV}
For $n\geq2$, if $m<n/2$, the law of \eqref{eq:process} is supported on a Haar-null compact subset of $\OO(n)$. Therefore
\begin{equation}\label{eq:TV}
 \|\nu_{n,m}^{A,x}-H_n\|_{\TV}=1.
\end{equation}
This includes $m=0$ with deterministic initial condition.
\end{theorem}

For a fixed deterministic sequence $A=(A_t)_{t\geq0}$, define the sweep mixing time
\[
 t_n^A(\eps)=\min\{m\geq0:
 \sup_{x\in\OO(n)}W_{2,\rho}(\nu_{n,m}^{A,x},H_n)\leq\eps\}.
\]
This definition also covers the homogeneous walk by taking every $A_t=I$. The distance is nonincreasing in time: using the same random sweep in two chains preserves their pointwise $\rho$ distance, while a deterministic orthogonal factor is an isometry, and each operation preserves Haar measure.

\begin{corollary}[A logarithmic gap at fixed normalized accuracy]\label{cor:mixing}
For every $0<\eps<\sqrt2$, there are $c_\eps>0$ and $n_\eps<\infty$ such that for all $n\geq n_\eps$ and every fixed deterministic insertion sequence,
\begin{equation}\label{eq:mixinggap}
 c_\eps\frac n{\log n}
 \leq t_n^A(\eps)
 \leq\left\lceil n\log\frac\pi\eps\right\rceil.
\end{equation}
The lower bound is also valid when, for each proposed total number of sweeps, the location of the single Fourier transform is chosen to minimize the distance. The upper bound holds at every location.
\end{corollary}

The corresponding absolute-accuracy mixing time is
\[
 t_{n,\mathrm F}^{A}(\eps)=\min\{m\geq0:
 \sup_{x\in\OO(n)}W_{2,\dF}(\nu_{n,m}^{A,x},H_n)\leq\eps\}.
\]
For every $\eps>0$, \eqref{eq:Wscaling} gives exactly
\begin{equation}\label{eq:mixingtimescaling}
 t_{n,\mathrm F}^{A}(\eps)=t_n^A(\eps/\sqrt n),
 \qquad t_n^A(\eps)=t_{n,\mathrm F}^{A}(\eps\sqrt n).
\end{equation}
The next conclusion uses the support bound at shrinking normalized radii, not just Corollary~\ref{cor:mixing}, whose constants were for fixed normalized accuracy.

\begin{corollary}[A logarithmic gap at fixed absolute accuracy]\label{cor:absolutemixing}
There is a universal constant $c_*>0$, for example $c_*=1/100$, with the following property. For every fixed $\eps>0$ there exists $n_\eps<\infty$ such that for all $n\geq n_\eps$ and all fixed deterministic insertion sequences,
\begin{equation}\label{eq:absolutemixinggap}
 c_*n\leq t_{n,\mathrm F}^{A}(\eps)
 \leq\left\lceil n\log\frac{\pi\sqrt n}{\eps}\right\rceil.
\end{equation}
More precisely, the proof shows $W_{1,\dF}(\nu_{n,m}^{A,x},H_n)>\eps$ for every $1\leq m\leq c_*n$ and every $x$ when $n$ is sufficiently large. The lower bound is uniform over insertions selected in advance for each proposed total sweep count, including optimization of the single Fourier factor's location. The upper bound holds at every location.
\end{corollary}

Thus the established fixed-accuracy ranges are $[\Omega(n/\log n),O(n)]$ in normalized distance and $[\Omega(n),O(n\log n)]$ in absolute distance. Both leave a logarithmic factor between lower and upper bounds; neither identifies a ratio between the two actual mixing times.

We prove the upper bound in Sections~\ref{sec:geometry}--\ref{sec:upperproof}, the lower bounds in Sections~\ref{sec:cover}--\ref{sec:lowerproof}, and singularity in Section~\ref{sec:singular}. The absolute-accuracy proof is in Section~\ref{subsec:absoluteproof}. Section~\ref{sec:extensions} treats determinant correction and ordered-tree sweeps.

\section{Tangent coordinates and the angle-derivative frame}\label{sec:geometry}

A probabilistic coupling of two nearby orthogonal matrices is easiest to describe in infinitesimal coordinates. A differentiable path through the identity in $\OO(n)$ has derivative $A$ satisfying $A^{\mathsf T}=-A$. We therefore write
\[
 \so(n)=\{A\in\R^{n\times n}:A^{\mathsf T}=-A\}.
\]
This is simply the vector space of real skew-symmetric matrices; its dimension is $n(n-1)/2$. We measure such infinitesimal displacements with
\begin{equation}\label{eq:ip}
 \ip AB=\tfrac12\tr(A^{\mathsf T}B),\qquad
 \norms A^2=\tfrac12\normF A^2.
\end{equation}
The factor $1/2$ is convenient because the elementary generators $E_{ij}$ then form an orthonormal basis, and $\norms A^2=\sum_{i<j}A_{ij}^2$. If $U$ is orthogonal, the change of coordinates $A\mapsto UAU^{\mathsf T}$ preserves this norm.

The same tangent norm defines a path distance $D$ within each determinant component. Concretely, a smooth path $\gamma$ has length
$\int\norms{\gamma(t)^{-1}\gamma'(t)}\,\dd t$, and $D(U,V)$ is the infimum of these lengths over paths joining $U$ to $V$. Left or right multiplication by a fixed orthogonal matrix does not change path lengths; this is the only meaning of ``bi-invariant'' that we use. For equal-determinant $U,V$,
\begin{equation}\label{eq:metriccomparison}
 \rho(U,V)\leq\sqrt{\frac2n}\,D(U,V),
 \qquad D(U,V)\leq\Delta_n:=\pi\sqrt{\lfloor n/2\rfloor}.
\end{equation}
The first inequality follows by bounding the Euclidean chord by the length of a curve. For the second, the real spectral decomposition of $UV^{-1}\in\SO(n)$ gives a skew-symmetric logarithm with at most $\lfloor n/2\rfloor$ rotation angles in $[-\pi,\pi]$. The exponential path of that logarithm has length at most $\Delta_n$.

We will also use the standard local identity
\[
 D(I,\exp A)=\norms A
\]
for sufficiently small skew-symmetric $A$. Intuitively, $t\mapsto\exp(tA)$ is the constant-speed shortest path starting at the identity in direction $A$ when $A$ is small. This familiar local fact for the orthogonal group is the only differential-geometric input needed below; translation gives the same statement near every point of either determinant component.

The sweep has two complementary properties. Its angle derivatives, transported to the identity, are orthogonal for every fixed angle vector. Their squared projection coefficients become isotropic after averaging over the angles and the coordinate permutation, even conditional on permutation parity. The first property will identify the displacement removed by the coupling; the second will determine its mean squared size.

Write $E_\ell=E_{\ell,\ell+1}$, and define partial sweeps and their angle directions by
\begin{equation}\label{eq:Bdef}
 S_0=I,\qquad S_\ell=e^{\theta_\ell E_\ell}\cdots e^{\theta_1E_1},
 \qquad B_\ell=S_{\ell-1}^{-1}E_\ell S_{\ell-1}.
\end{equation}
In particular, $B_\ell$ depends only on $\theta_1,\ldots,\theta_{\ell-1}$, not on the current or later angles. This dependence on preceding angles will also make the coupling measure preserving.

The quantity $B_\ell$ has a simple operational interpretation. The derivative $\partial_{\theta_\ell}Q$ is a tangent vector at the current matrix $Q(\theta)$. Multiplying it on the left by $Q(\theta)^{-1}$ transports that tangent vector back to the identity, where tangent vectors can be identified with ordinary skew-symmetric matrices. Lemma~\ref{lem:frame} shows that these transported derivatives are exactly the matrices $B_\ell$, and that they are mutually orthogonal.

\begin{lemma}[Orthogonal frame]\label{lem:frame}
For every angle vector,
\begin{equation}\label{eq:derivative}
 Q(\theta)^{-1}\frac{\partial Q}{\partial\theta_\ell}(\theta)=B_\ell,
 \qquad \ip{B_\ell}{B_j}=\mathbf1_{\{\ell=j\}}.
\end{equation}
Consequently, for every $v\in\R^{n-1}$,
\begin{equation}\label{eq:flatdifferential}
 \normF{\dd Q_\theta(v)}^2=2\|v\|_2^2.
\end{equation}
\end{lemma}
\begin{proof}
Write $R_j=e^{\theta_jE_j}$ and
$T_\ell=R_{n-1}\cdots R_{\ell+1}$, with an empty product equal to $I$.
Then $Q=T_\ell R_\ell S_{\ell-1}$, so differentiation and cancellation give
\[
 \begin{aligned}
 Q^{-1}\partial_{\theta_\ell}Q
 &=S_{\ell-1}^{-1}R_\ell^{-1}T_\ell^{-1}
       T_\ell E_\ell R_\ell S_{\ell-1}\\
 &=S_{\ell-1}^{-1}E_\ell S_{\ell-1}=B_\ell.
 \end{aligned}
\]
Only the commutation of $E_\ell$ with its own exponential is used; different rotation factors have not been interchanged.

Set $u_\ell=S_{\ell-1}^{-1}e_\ell$. The prefix $S_{\ell-1}$ acts only on the first $\ell$ coordinates and fixes $e_{\ell+1}$. Thus
\begin{equation}\label{eq:wedge}
 B_\ell=u_\ell e_{\ell+1}^{\mathsf T}
 -e_{\ell+1}u_\ell^{\mathsf T},\qquad
 u_\ell\in\operatorname{span}(e_1,\ldots,e_\ell),\quad \|u_\ell\|=1.
\end{equation}
Equivalently,
\[
 B_\ell=\sum_{a=1}^{\ell}(u_\ell)_aE_{a,\ell+1}.
\]
The basis elements appearing in this sum have larger index $\ell+1$, whereas those appearing in $B_j$ have larger index $j+1$. For $j\ne\ell$ the two sets are disjoint. Orthonormality of $(E_{ab})_{a<b}$ therefore gives
\[
 \ip{B_\ell}{B_j}=0\quad(j\ne\ell),\qquad
 \norms{B_\ell}^2=\sum_{a=1}^{\ell}(u_\ell)_a^2=1.
\]
This is pointwise orthogonality, not merely orthogonality in expectation.
Finally, $\dd Q_\theta(v)=Q\sum_\ell v_\ell B_\ell$, whence
\[
 \normF{\dd Q_\theta(v)}^2
 =2\norms{\sum_\ell v_\ell B_\ell}^{\!2}
 =2\sum_\ell v_\ell^2,
\]
proving \eqref{eq:flatdifferential}.
\end{proof}

\begin{lemma}[Diagonal covariance of the frame coefficients]\label{lem:cov}
For independent uniform angles and every $\ell$,
\begin{equation}\label{eq:cov}
 \E[(u_\ell)_a(u_\ell)_b]=0\quad(a\ne b),\qquad
 \sum_{a=1}^{\ell}\E[(u_\ell)_a^2]=1.
\end{equation}
\end{lemma}
\begin{proof}
We have $u_1=e_1$. For $\ell\geq2$,
$S_{\ell-1}=R_{\ell-1,\ell}(\theta_{\ell-1})S_{\ell-2}$.
The inverse rotation sends $e_\ell$ to
$-\sin\theta_{\ell-1}\,e_{\ell-1}+\cos\theta_{\ell-1}\,e_\ell$,
and $S_{\ell-2}$ fixes $e_\ell$. Hence
\[
 u_\ell=-\sin\theta_{\ell-1}\,u_{\ell-1}
          +\cos\theta_{\ell-1}\,e_\ell.
\]
Let $M_\ell=\E[u_\ell u_\ell^{\mathsf T}]$, regarding all these matrices as $n\times n$ matrices with zeros on unused coordinates. Since $u_{\ell-1}$ depends only on angles preceding $\theta_{\ell-1}$, expansion of the outer product gives
\[
 \begin{aligned}
 M_\ell
 &=\E[\sin^2\theta_{\ell-1}]M_{\ell-1}
   +\E[\cos^2\theta_{\ell-1}]e_\ell e_\ell^{\mathsf T}\\
 &\quad-\E[\sin\theta_{\ell-1}\cos\theta_{\ell-1}]
       \E[u_{\ell-1}e_\ell^{\mathsf T}+e_\ell u_{\ell-1}^{\mathsf T}]\\
 &=\tfrac12M_{\ell-1}+\tfrac12e_\ell e_\ell^{\mathsf T}.
 \end{aligned}
\]
Starting from $M_1=e_1e_1^{\mathsf T}$, induction shows that $M_\ell$ is diagonal. Its trace is one because $\|u_\ell\|=1$ pointwise. These are precisely the two assertions in \eqref{eq:cov}. Here the required matrix is the uncentered second moment; its diagonal entries need not be equal.
\end{proof}

The next calculation explains exactly what the random permutation contributes. For a fixed sweep, the $n-1$ angle directions occupy only a small part of the full tangent space. Averaging over the permutation redistributes those directions among the coordinate planes so that, in mean square, every coefficient of a skew-symmetric displacement is treated equally. The conclusion is an exact isotropy identity, and it remains true even after conditioning on whether the permutation is even or odd.

\begin{lemma}[Isotropic projection, including fixed permutation parity]\label{lem:isotropy}
Let $A\in\so(n)$. Let $P$ be uniform either over all permutation matrices or over the permutation matrices with a prescribed determinant $\tau\in\{-1,+1\}$. Independently let the angles be uniform. Then
\begin{equation}\label{eq:isotropy}
 \E\sum_{\ell=1}^{n-1}\ip{PAP^{-1}}{B_\ell}^{\!2}
 =\frac2n\norms A^2.
\end{equation}
\end{lemma}
\begin{proof}
First fix $P$ and put $A'=PAP^{-1}$. By \eqref{eq:wedge},
$\ip{A'}{B_\ell}=\sum_{a\leq\ell}(u_\ell)_aA'_{a,\ell+1}$.
Squaring before averaging gives
\[
 \begin{aligned}
 \E_\theta\ip{A'}{B_\ell}^{\!2}
 &=\sum_{a,b\leq\ell}
   \E[(u_\ell)_a(u_\ell)_b]A'_{a,\ell+1}A'_{b,\ell+1}\\
 &=\sum_{a\leq\ell}w_{\ell,a}(A'_{a,\ell+1})^2,
 \qquad w_{\ell,a}=\E[(u_\ell)_a^2],\quad
 \sum_a w_{\ell,a}=1.
 \end{aligned}
\]
The cross terms vanish by Lemma~\ref{lem:cov}. This is the role of angle averaging; permutation symmetry alone is not the justification for discarding those terms.

For distinct fixed $a,b$, a uniform permutation sends their preimages to a uniform ordered pair of distinct indices. Since $A_{ji}^2=A_{ij}^2$,
\begin{equation}\label{eq:permaverage}
 \E_P(A'_{ab})^2
 =\frac1{n(n-1)}\sum_{i\ne j}A_{ij}^2
 =\frac{\norms A^2}{\binom n2}.
\end{equation}
To obtain the same identity conditional on parity, let $T_{ab}$ transpose coordinates $a,b$. The map $P\mapsto T_{ab}P$ is a bijection between the two parity classes, and
\[
 \bigl((T_{ab}P)A(T_{ab}P)^{-1}\bigr)_{ab}
 =(T_{ab}A'T_{ab}^{-1})_{ab}=A'_{ba}=-A'_{ab}.
\]
It therefore preserves the square in \eqref{eq:permaverage}. The two conditional means are equal, and each equals their unconditional mean. This argument does not require an additional transitivity assertion about the alternating group.

For either allowed law of $P$, independence of the angles and $P$ now yields
\[
 \E_{P,\theta}\ip{PAP^{-1}}{B_\ell}^{\!2}
 =\sum_{a\leq\ell}w_{\ell,a}
       \frac{\norms A^2}{\binom n2}
 =\frac{\norms A^2}{\binom n2}.
\]
Summing over $\ell$ and using $(n-1)/\binom n2=2/n$ proves \eqref{eq:isotropy}. Only linearity of expectation is used in this final summation; the directions $B_\ell$ are generally dependent.
\end{proof}

By Lemma~\ref{lem:frame}, the sum on the left of \eqref{eq:isotropy} is the squared norm of the orthogonal projection of $PAP^{-1}$ onto $\operatorname{span}\{B_1,\ldots,B_{n-1}\}$. Thus its expected fraction of the total squared norm is $(n-1)/\dim\so(n)=2/n$. This averaged projection identity, rather than a uniform distribution of the random subspace, is the isotropy needed below.

\section{The coupling and the Wasserstein upper bound}\label{sec:upperproof}

The goal is to couple two copies started at nearby matrices. We use the same random permutation in both copies and make a small, carefully chosen change to the angles in the second copy. To first order, this change cancels the component of the initial displacement lying in the $n-1$ angle directions identified in the previous section. The only subtle point is probabilistic: after changing the angles, the second angle vector must still have exactly the product-uniform distribution required by the Markov kernel. The triangular construction below guarantees this exactly, not approximately.

\subsection{A measure-preserving triangular change of angles}

We adapt the angle-shift projection coupling of
\cite[Section~4.3, proof of Lemma~1]{Oliveira09} to an entire ordered sweep.
A separate shift in each angle can remove the corresponding tangent projection, but these shifts depend on the other angles. The dependence on preceding angles in \eqref{eq:Bdef} allows all shifts to be made simultaneously without changing the product-uniform law. We first record the elementary change-of-variables fact that justifies this step.

\begin{lemma}[Triangular torus shifts]\label{lem:triangular}
Suppose $a_\ell:\T^{\ell-1}\to\R$ is a smooth periodic function for $1\leq\ell\leq s$, with $a_1$ constant. For every real $h$, the map
\begin{equation}\label{eq:torusshift}
 \theta'_\ell=\theta_\ell-h a_\ell(\theta_1,\ldots,\theta_{\ell-1})
 \pmod{2\pi}
\end{equation}
is a smooth measure-preserving bijection of $\T^s$. In particular, independent uniform input angles give independent uniform output angles.
\end{lemma}
\begin{proof}
Periodicity makes the map well defined and smooth on the torus, independently of the chosen angular representatives. Its inverse is obtained successively from
\[
 \theta_1=\theta'_1+ha_1\pmod{2\pi},\qquad
 \theta_\ell=\theta'_\ell+
 h a_\ell(\theta_1,\ldots,\theta_{\ell-1})\pmod{2\pi}.
\]
At each stage every argument on the right has already been recovered, so the inverse is unique and smooth. In local angular coordinates the Jacobian matrix is lower triangular with diagonal entries equal to one. Its determinant is one, and the change-of-variables formula proves preservation of normalized Lebesgue measure on $\T^s$.

To see explicitly why joint independence is preserved, let
$\mathcal F_{\ell-1}=\sigma(\theta_1,\ldots,\theta_{\ell-1})$.
For every Borel set $C\subset\T$,
\[
 \PP\{\theta'_\ell\in C\mid\mathcal F_{\ell-1}\}
 =\lambda(C),
\]
where $\lambda$ is uniform probability measure on $\T$. Indeed the shift is fixed conditional on $\mathcal F_{\ell-1}$, whereas $\theta_\ell$ remains uniform. The transformed preceding angles are $\mathcal F_{\ell-1}$-measurable, so conditioning again shows that $\theta'_\ell$ is uniform independently of them. Induction gives the product-uniform law. No smallness assumption on $h$ is needed for this assertion.
\end{proof}

The orthogonal group has two connected components, distinguished by the determinant: $\OO(n)_+=\SO(n)$ and $\OO(n)_-=\{U:\det U=-1\}$. A rotation sweep $Q(\theta)$ has determinant $+1$, so the only part of one sweep that can switch components is the permutation. Thus an even permutation preserves the determinant component and an odd permutation swaps the two components.

For $\sigma,\tau\in\{-1,+1\}$, let $K_\tau$ be the sweep kernel conditioned on $\det P=\tau$. It maps $\OO(n)_\sigma$ into $\OO(n)_{\tau\sigma}$. Both components carry the distance $D$.

\begin{lemma}[Infinitesimal contraction]\label{lem:local}
For $n\geq3$ and either parity $\tau$, the following estimate holds uniformly over equal-determinant $X,Y$ with $h=D(X,Y)$ sufficiently small:
\begin{equation}\label{eq:local}
 W_{2,D}(K_\tau(X,\cdot),K_\tau(Y,\cdot))^2
 \leq r_n^2h^2+C_n h^3,
\end{equation}
where $C_n<\infty$ may depend on $n$ but not on $X,Y$.
\end{lemma}
\begin{proof}
\emph{Construction and validity of the coupling.}
The case $X=Y$ follows by synchronous coupling. Otherwise, because $X$ and $Y$ are sufficiently close and have the same determinant, there is a skew-symmetric matrix $A$ with $\norms A=1$ such that
$Y=e^{hA}X$.
Use the same permutation $P$ in both chains, chosen uniformly with parity $\tau$, and put $A'=PAP^{-1}$. Draw the first angle vector $\theta$ independently of $P$ with product-uniform law. For the second chain set
\begin{equation}\label{eq:couplingangles}
 \theta'_\ell=\theta_\ell-h\ip{A'}{B_\ell(\theta_1,\ldots,\theta_{\ell-1})}
 \pmod{2\pi}.
\end{equation}
For each fixed $P$, the shift is a smooth periodic function only of preceding original angles. Lemma~\ref{lem:triangular} therefore gives the exact product-uniform law of $\theta'$ conditional on $P$. Since this conditional law is the same for every $P$, the vector $\theta'$ is also independent of $P$. Thus
$X'=Q(\theta)PX$ and $Y'=Q(\theta')PY$ have the required transition laws, with no approximation of either marginal. Both have determinant $\tau\det X$.

\emph{First-order displacement.}
The identity $Pe^{hA}=e^{hA'}P$ gives
$Y'=Q(\theta')e^{hA'}PX$. Right multiplying both outputs by $(PX)^{-1}$ and then left multiplying by $Q(\theta)^{-1}$ yields
\begin{equation}\label{eq:relativeoutput}
 D(X',Y')=D\bigl(I,Q(\theta)^{-1}Q(\theta')e^{hA'}\bigr).
\end{equation}
Write $a_\ell=\ip{A'}{B_\ell(\theta)}$ and $a=(a_1,\ldots,a_{n-1})$. With $P,A,\theta$ fixed, introduce
\[
 M(u)=Q(\theta)^{-1}Q(\theta-ua)e^{uA'}.
\]
The coefficients $a_\ell$ are evaluated at the original $\theta$ and held fixed as $u$ varies; they are not recomputed at $\theta-ua$. Consequently Lemma~\ref{lem:frame} gives
\[
 M(0)=I,\qquad
 M'(0)=A'-\sum_{\ell=1}^{n-1}a_\ell B_\ell=:Z.
\]

\emph{Projection identity.}
For each realization, $\sum_\ell a_\ell B_\ell$ is the orthogonal projection of $A'$ onto the span of the frame. Explicitly, orthonormality and the choice $a_\ell=\ip{A'}{B_\ell}$ give
\begin{equation}\label{eq:residual}
 \begin{aligned}
 \norms Z^2
 &=\norms{A'}^2
   -2\sum_\ell a_\ell\ip{A'}{B_\ell}
   +\sum_{\ell,j}a_\ell a_j\ip{B_\ell}{B_j}\\
 &=\norms{A'}^2-\sum_{\ell=1}^{n-1}a_\ell^2.
 \end{aligned}
\end{equation}
The negative sum includes the cross term with $A'$; orthogonality eliminates only cross terms between distinct frame vectors. In particular, $\|a\|_2\leq1$ and $\norms Z\leq1$, since $\norms{A'}=1$.

\emph{Uniform remainder and averaging.}
For fixed $n$, the unit sphere of $\so(n)$ and the angle torus are compact, the permutation set is finite, and $M(u)$ is smooth in all continuous parameters. As $M(0)=I$ for every choice of those parameters, a sufficiently small interval in $u$ keeps all $M(u)$ in a common neighborhood of $I$ on which the matrix logarithm is well defined and smooth. Taylor's theorem therefore gives
\[
 \log M(h)=hZ+R_h,\qquad \norms{R_h}\leq c_n h^2,
\]
with $c_n$ and the interval independent of $P,A,\theta$, and hence of the starting pair. Periodicity of $Q$ ensures that reducing angles modulo $2\pi$ introduces no discontinuity in this argument. Using \eqref{eq:relativeoutput} and the local distance identity,
\[
 \begin{aligned}
 D(X',Y')^2
 &=\norms{hZ+R_h}^2\\
 &=h^2\norms Z^2+2h\ip Z{R_h}+\norms{R_h}^2
 =h^2\norms Z^2+O_n(h^3),
 \end{aligned}
\]
where the last remainder is uniform because $\norms Z\leq1$.
Only now take expectations. By \eqref{eq:residual} and the parity-conditioned identity of Lemma~\ref{lem:isotropy},
\[
 \E\norms Z^2=1-\E\sum_{\ell=1}^{n-1}a_\ell^2
 =1-\frac2n=r_n^2.
\]
Thus the expected squared cost of this admissible coupling is at most
$r_n^2h^2+C_nh^3$. Taking the infimum over couplings proves \eqref{eq:local}.
\end{proof}

\subsection{From nearby starting points to probability measures}

We have so far shown contraction only for two sufficiently nearby starting points. Oliveira's local-to-global theorem is the device that removes this restriction: on a length space, a uniform infinitesimal Wasserstein contraction can be integrated along a shortest path to give the same contraction factor for arbitrary starting points, and then for arbitrary probability measures. A fixed output isometry lets us apply that theorem on a single determinant component even when the permutation parity is odd. We use \cite[Theorem~3]{Oliveira09} in exactly this way.

\begin{lemma}[Componentwise contraction via Oliveira's theorem]
\label{lem:global}
Let $n\geq3$ and $\sigma,\tau\in\{-1,+1\}$. For any probability
measures $\alpha,\beta$ supported on $\OO(n)_\sigma$,
\begin{equation}\label{eq:global}
 W_{2,D}(\alpha K_\tau,\beta K_\tau)
 \leq r_n W_{2,D}(\alpha,\beta),
 \qquad r_n=\sqrt{1-\frac2n}.
\end{equation}
The distance on the left is computed within $\OO(n)_{\tau\sigma}$,
and the distance on the right within $\OO(n)_\sigma$.
\end{lemma}

\begin{proof}
Fix $\sigma,\tau$ and put $M=\OO(n)_\sigma$. Define
\[
 J_+=I,\qquad
 J_-=\operatorname{diag}(-1,1,\ldots,1),\qquad
 T_\tau(U)=J_\tau U.
\]
Since $\det J_\tau=\tau$, the map $T_\tau$ sends
$\OO(n)_{\tau\sigma}$ bijectively onto $M$. Left multiplication
preserves the intrinsic distance $D$, so $T_\tau$ is an isometry.
We write $(T_\tau)_\#\eta$ for the distribution of $T_\tau(U)$ when $U$ has law $\eta$. Consequently, for any probability measures $\eta,\zeta$ supported
on $\OO(n)_{\tau\sigma}$,
\[
 W_{2,D}\bigl((T_\tau)_\#\eta,(T_\tau)_\#\zeta\bigr)
 =W_{2,D}(\eta,\zeta).
\]
Indeed, pushing both coordinates of a coupling forward by $T_\tau$
preserves its cost, and the inverse isometry gives the reverse
inequality.

Define a Markov kernel on $M$ by
\[
 \widehat K_{\sigma,\tau}(X,\cdot)
 :=(T_\tau)_\#K_\tau(X,\cdot),
 \qquad X\in M.
\]
The space $(M,D)$ is a compact metric space in which the distance between two points is realized by path length; in particular it is a Polish length space in the terminology of Oliveira's theorem, and all probability measures on it have finite second moments. The transformed kernel is also Feller, meaning that it sends continuous test functions to continuous functions: for every continuous $f$ on $M$,
\[
 \widehat K_{\sigma,\tau}f(X)
 =\mathbb E\!\left[
     f\bigl(J_\tau Q(\theta)PX\bigr)
     \,\middle|\,\det P=\tau
   \right]
\]
is continuous in $X$ by dominated convergence.

By the isometry identity and Lemma~\ref{lem:local}, whenever
$X,Y\in M$ and $h=D(X,Y)>0$ is sufficiently small,
\[
 \frac{
 W_{2,D}\bigl(
   \widehat K_{\sigma,\tau}(X,\cdot),
   \widehat K_{\sigma,\tau}(Y,\cdot)
 \bigr)}{D(X,Y)}
 \leq \sqrt{r_n^2+C_nh}.
\]
Thus, for every $X\in M$,
\[
 \limsup_{\substack{Y\to X\\Y\ne X}}
 \frac{
 W_{2,D}\bigl(
   \widehat K_{\sigma,\tau}(X,\cdot),
   \widehat K_{\sigma,\tau}(Y,\cdot)
 \bigr)}{D(X,Y)}
 \leq r_n.
\]
All hypotheses of the theorem
\cite[Theorem~3]{Oliveira09} are therefore satisfied, with $p=2$
and Lipschitz constant $r_n$. It follows that
\[
 W_{2,D}\bigl(
   \alpha\widehat K_{\sigma,\tau},
   \beta\widehat K_{\sigma,\tau}
 \bigr)
 \leq r_n W_{2,D}(\alpha,\beta).
\]
Finally,
$\alpha\widehat K_{\sigma,\tau}=(T_\tau)_\#(\alpha K_\tau)$,
and likewise for $\beta$. Undoing the output isometry gives
\begin{equation}\label{eq:conditionalcontraction}
 \begin{aligned}
 W_{2,D}(\alpha K_\tau,\beta K_\tau)
 &=
 W_{2,D}\bigl(
   \alpha\widehat K_{\sigma,\tau},
   \beta\widehat K_{\sigma,\tau}
 \bigr)\\
 &\leq r_n W_{2,D}(\alpha,\beta),
 \end{aligned}
\end{equation}
as claimed. The limit above is taken with $n$ fixed, so no bound
on $C_n$ uniform in $n$ is required.
\end{proof}

\subsection{Matching determinant components and counting every sweep}

If $\alpha,\beta$ have the same masses $w_+,w_-$ on the two components, define
\begin{equation}\label{eq:Wcomponent}
 \Wc(\alpha,\beta)^2
 =\sum_{\sigma\in\{-1,+1\}}w_\sigma
 W_{2,D}(\alpha_\sigma,\beta_\sigma)^2,
\end{equation}
where $\alpha_\sigma,\beta_\sigma$ are the conditional laws and zero-weight terms are omitted. Equivalently, this is the least expected squared intrinsic distance among couplings that match the determinants almost surely. By \eqref{eq:metriccomparison},
\begin{equation}\label{eq:Wcomparison}
 W_{2,\rho}(\alpha,\beta)\leq\sqrt{2/n}\,\Wc(\alpha,\beta).
\end{equation}

\begin{lemma}[Full-kernel component-matched contraction]\label{lem:components}
Let $K=\tfrac12(K_++K_-)$ be the unconditioned sweep kernel. For all $\alpha,\beta$ with equal component masses,
\begin{equation}\label{eq:componentcontraction}
 \Wc(\alpha K,\beta K)\leq r_n\Wc(\alpha,\beta).
\end{equation}
Left multiplication by any fixed orthogonal matrix preserves $\Wc$.
\end{lemma}
\begin{proof}
Both output laws have component masses $1/2$. Conditional on output determinant $\omega$, the first law is
$\sum_\sigma w_\sigma\alpha_\sigma K_{\omega\sigma}$, and the second has the corresponding expression with $\beta$. Squared Wasserstein distance is jointly convex under mixing with common weights: mix an optimal coupling for each pair of summands. Applying \eqref{eq:conditionalcontraction}, the squared distance between these output conditional laws is at most
$\sum_\sigma w_\sigma r_n^2W_{2,D}(\alpha_\sigma,\beta_\sigma)^2$.
Averaging this bound over the two values of $\omega$ proves \eqref{eq:componentcontraction}. A fixed left multiplier is an isometry on each component and permutes the two components, proving the last assertion.
\end{proof}

\begin{proof}[Proof of Theorem~\ref{thm:upper}]
Let $\sigma=\det X_0$ and compare $\alpha_0=\delta_{X_0}$ with $\beta_0=H_{n,\sigma}$. Their determinants match and
$\Wc(\alpha_0,\beta_0)\leq\Delta_n$.
Apply to both laws the kernels and fixed multipliers in \eqref{eq:process}. Left multiplication by any matrix of determinant $\tau$ sends $H_{n,\sigma}$ to $H_{n,\tau\sigma}$. Therefore $\beta_0 K=H_n$, since permutation parity is fair, and after every later step the comparison law remains $H_n$. Lemma~\ref{lem:components} applies at every sweep, including the first, and gives
\[
 \Wc(\nu_{n,m}^{A,x},H_n)\leq\Delta_n r_n^m.
\]
Combine this with \eqref{eq:Wcomparison} to obtain the $b_n r_n^m$ bound. The inequality
$\frac12\log(1-2/n)\leq-1/n$ gives the exponential form; rescaling the metric gives \eqref{eq:unnormalized}.

For the bound by $\sqrt2$, independently couple any $V\sim\nu$ with $U\sim H_n$. Each Haar matrix entry has mean zero, by invariance under row sign changes. Hence
\begin{equation}\label{eq:independentbound}
 \E\rho(U,V)^2=2-\frac2n\E\tr(V^{\mathsf T}U)=2.
\end{equation}
This proves $W_{2,\rho}(\nu,H_n)\leq\sqrt2$ for every $\nu$.
Finally, for $n=2$ a single uniform plane angle is Haar on $\SO(2)$, and the independent uniform parity chooses the two components equally. Deterministic orthogonal factors preserve Haar measure, so the claim follows.
\end{proof}

\section{Metric entropy of the sampler's support}\label{sec:cover}

The lower bound starts from a simple observation: after the permutation choices are fixed, $m$ sweeps depend on only $m(n-1)$ continuous angles. We quantify how many normalized Frobenius balls are needed to cover all possible outputs. Later we compare this covering number with the exponentially small Haar mass of a ball of fixed radius.

Let $S_{n,m}^{A,x}$ be the set of outputs in \eqref{eq:process} as every angle and permutation varies. It is compact, since it is a finite union of continuous images of a compact torus. It supports $\nu_{n,m}^{A,x}$.
For a set $S$ and $\delta>0$, write $N_\rho(S,\delta)$ for the minimum cardinality of a $\delta$-net with centers in $S$. Its logarithm is often called the metric entropy of $S$ at scale $\delta$.

\begin{lemma}[Covering a sweep product]\label{lem:cover}
For $m\geq1$ and $\delta>0$, let
\[
 q(m,\delta)=\left\lceil\frac{\pi\sqrt2\,m}{\delta}\right\rceil.
\]
Then
\begin{equation}\label{eq:covergeneral}
 N_\rho(S_{n,m}^{A,x},\delta)
 \leq(n!)^m q(m,\delta)^{m(n-1)}.
\end{equation}
In particular,
\begin{equation}\label{eq:coverquarter}
 N_\rho(S_{n,m}^{A,x},1/4)\leq(n!)^m(40m)^{m(n-1)}.
\end{equation}
\end{lemma}
\begin{proof}
Fix the permutation tuple. If each coordinate of two angle vectors is at circular distance at most $a$, choose coordinatewise shortest lifts and interpolate along the straight segment. Lemma~\ref{lem:frame} bounds the normalized length of its image by
\[
 \sqrt{\frac2n}\,\|\theta-\phi\|_2
 \leq\sqrt{\frac{2(n-1)}n}\,a\leq\sqrt2\,a.
\]
Thus this is an upper bound for the $\rho$ distance between the two sweeps. If two full angle arrays differ by at most $a$ in every coordinate, telescope the product of the $m$ sweeps. Every left and right multiplier in each telescoping term is orthogonal. The triangle inequality therefore bounds the output distance by $m\sqrt2\,a$, regardless of the fixed factors $A_t$.

A grid of $q$ equally spaced angles on $\T$ has covering radius $\pi/q$. Take $q=q(m,\delta)$ and grid each of the $m(n-1)$ angle coordinates. The resulting output grid is a $\delta$-net for this permutation branch. There are $(n!)^m$ branches, proving \eqref{eq:covergeneral}. For $\delta=1/4$,
$q\leq4\pi\sqrt2\,m+1<40m$, which proves \eqref{eq:coverquarter}.
\end{proof}

For fixed $\delta>0$ and $m=o(n/\log n)$, this yields the uniform estimate
\begin{equation}\label{eq:entropyasymptotic}
 \log N_\rho(S_{n,m}^{A,x},\delta)
 \leq m\log(n!)+m(n-1)\log q(m,\delta)=o(n^2).
\end{equation}
Indeed, eventually $m\leq n$, so both logarithmic factors are $O_\delta(\log n)$ after extracting their respective factors of $n$.

\section{Elementary small-ball bounds for Haar measure}\label{sec:smallball}

The covering estimate above is useful only if a typical Haar matrix is unlikely to lie near any one net point. The next lemmas provide that complementary estimate. The proof uses only the familiar sequential construction of a Haar orthogonal matrix: conditional on the first few columns, the next column is uniform on the unit sphere in the remaining orthogonal complement.

The Haar small-ball estimate below is an explicit variant of the
argument in \cite[Section~6.1, Claim~1 and its proof]{Oliveira09}.
We retain a free diagonal threshold to cover every fixed normalized
radius below $\sqrt2$, track the finite-dimensional constants, and
include both determinant components. We first record the elementary
spherical-cap estimate used in this adaptation.

\begin{lemma}[Spherical cap]\label{lem:cap}
Let $Z$ be uniform on the unit sphere of $\R^d$, with $d\geq2$, and let $v\in\R^d$ satisfy $\|v\|\leq1$. For $0<a<1$,
\begin{equation}\label{eq:cap}
 \PP\{\langle Z,v\rangle\geq a\}
 \leq\frac12(1-a^2)^{(d-1)/2}.
\end{equation}
\end{lemma}
\begin{proof}
By rotation invariance and monotonicity in $\|v\|$, it is enough to take $v=e_1$. Write $Z=G/\|G\|$, where the coordinates of $G$ are independent standard Gaussians, and let $S=\sum_{j=2}^dG_j^2$. The event is
\[
 G_1\geq\frac{a}{\sqrt{1-a^2}}\sqrt S.
\]
For a standard Gaussian, $\overline\Phi(t)\leq\frac12e^{-t^2/2}$ for $t\geq0$. For completeness, the difference between the right side and the left side vanishes at zero and at infinity; its derivative is
$e^{-t^2/2}(1/\sqrt{2\pi}-t/2)$, which changes sign exactly once from positive to negative. The difference is therefore nonnegative.
Conditioning on $S$ and integrating gives
\[
 \PP\{Z_1\geq a\}
 \leq\frac12\E\exp\left\{-\frac{a^2S}{2(1-a^2)}\right\}
 =\frac12(1-a^2)^{(d-1)/2},
\]
where the last identity is the product of $d-1$ elementary Gaussian integrals.
\end{proof}

We use the following familiar description of Haar columns, including a justification. For Haar measure on $\OO(n)$, conditional on the first $j-1$ columns, column $j$ is uniform on the unit sphere in their orthogonal complement. It follows either from invariance under orthogonal transformations fixing those columns, or by Gram--Schmidt orthogonalization of independent standard Gaussian columns. In the latter construction the output law is left orthogonally invariant and hence Haar. The same description holds on either prescribed determinant component for $j\leq n-1$: conditioning the orientation only chooses the sign of the final column and leaves the first $n-1$ columns unchanged. This last fact can also be seen by right multiplication by a diagonal matrix that reverses only the last column.

\begin{lemma}[Haar small-ball bound; variant of Oliveira's Claim~1]\label{lem:haarball}
Let $0<R<\sqrt2$, put $t=1-R^2/2$, and choose $0<a<t$. Define
\begin{equation}\label{eq:haarparameters}
 s=\frac{t-a}{1-a},\qquad k=\lfloor sn\rfloor,
 \qquad
 B_n(R,a)=\min\left\{1,\,
 2^{n-k}(1-a^2)^{k(2n-k-1)/4}\right\}.
\end{equation}
For any $V\in\OO(n)$,
\begin{equation}\label{eq:haarball}
 H_n\{U:\rho(U,V)\leq R\}\leq B_n(R,a).
\end{equation}
The same bound holds for either conditional Haar law $H_{n,\sigma}$ and any center $V\in\OO(n)$. For fixed $R,a$,
\begin{equation}\label{eq:ballasymptotic}
 \log B_n(R,a)\leq-c(R,a)n^2+O_{R,a}(n),\qquad
 c(R,a)=-\frac{s(2-s)}4\log(1-a^2)>0,
\end{equation}
for all sufficiently large $n$.
\end{lemma}
\begin{proof}
Left multiplication by $V^{\mathsf T}$ reduces the center to $I$ and sends any prescribed determinant component to a prescribed component. Since
$\rho(U,I)^2=2-2\tr(U)/n$, membership in the ball implies $\tr U\geq tn$.
Let $J=\{j:U_{jj}\geq a\}$. All diagonal entries are at most one, and those outside $J$ are at most $a$, so
\[
 tn\leq\tr U\leq |J|+(n-|J|)a.
\]
Consequently $|J|\geq sn$, and in particular $J$ contains a set of $k$ indices. If $k=0$, the stated bound is the trivial bound one, so assume $k\geq1$. We have $k\leq n-1$ because $s<1$.

For any fixed set of $k$ indices, conjugation by a permutation carries its diagonal entries to the first $k$ diagonal entries and preserves Haar measure on either component. It suffices to bound the probability that $U_{jj}\geq a$ for $1\leq j\leq k$.
Conditionally on the first $j-1$ columns, column $j$ is uniform on a sphere in dimension $d_j=n-j+1\geq2$. Its inner product with $e_j$ is its inner product with the orthogonal projection of $e_j$ into that subspace, a vector of norm at most one. Lemma~\ref{lem:cap} gives the conditional bound
\[
 \PP\{U_{jj}\geq a\mid U_{\cdot1},\ldots,U_{\cdot,j-1}\}
 \leq\frac12(1-a^2)^{(n-j)/2}.
\]
Iterated conditioning yields
\[
 \PP\{U_{jj}\geq a\text{ for all }j\leq k\}
 \leq2^{-k}(1-a^2)^{\frac12\sum_{j=1}^k(n-j)}
 =2^{-k}(1-a^2)^{k(2n-k-1)/4}.
\]
A union bound over at most $\binom nk\leq2^n$ sets proves \eqref{eq:haarball}. Finally, $k=sn+O(1)$, so the logarithm of the untruncated expression in \eqref{eq:haarparameters} equals
$\frac{s(2-s)}4\log(1-a^2)n^2+O_{R,a}(n)$. It is negative for all sufficiently large $n$, proving \eqref{eq:ballasymptotic}.
\end{proof}

At $R=1/2$ and $a=3/4$, we have $t=7/8$, $s=1/2$, and $1-a^2=7/16$. Thus Lemma~\ref{lem:haarball} gives exactly the value $B_n$ in \eqref{eq:Bn}. Its quadratic exponent is
\begin{equation}\label{eq:kappa}
 \kappa=\frac3{16}\log\frac{16}7>0.
\end{equation}

\section{Proofs of the Wasserstein lower bounds}\label{sec:lowerproof}

The argument now combines the two preceding ingredients. A small neighborhood of the sampler's support is covered by relatively few balls, while each such ball has very small Haar probability. Therefore most Haar matrices remain a definite distance from every possible output of the sampler, which forces a Wasserstein lower bound.

\begin{lemma}[Support-neighborhood separation]\label{lem:separation}
Let $(M,d)$ be a compact metric space, $H$ a probability measure, and $\nu$ supported on a compact set $S$. Suppose that $S$ has a $\delta$-net of cardinality at most $L$, and that every ball of radius $R>\delta$ centered at a net point has $H$ mass at most $B$. Then for $p\geq1$,
\begin{equation}\label{eq:separation}
 W_{p,d}(\nu,H)
 \geq(R-\delta)\bigl(1-\min\{1,LB\}\bigr)^{1/p}.
\end{equation}
\end{lemma}
\begin{proof}
Every point whose distance from $S$ is at most $R-\delta$ lies within distance $R$ of a net point. The $H$ mass of this neighborhood is at most $\min\{1,LB\}$. In any coupling $X\sim\nu$, $U\sim H$, the inequality $d(X,U)\geq\dist(U,S)$ holds almost surely. Therefore
\[
 \E d(X,U)^p\geq(R-\delta)^p
 H\{U:\dist(U,S)>R-\delta\}
 \geq(R-\delta)^p(1-\min\{1,LB\}).
\]
Taking the infimum over couplings proves the result. Equivalently for $p=1$, the $1$-Lipschitz function $U\mapsto\dist(U,S)$ is a witness to the discrepancy.
\end{proof}

\begin{proof}[Proof of Theorem~\ref{thm:lower}]
Apply Lemma~\ref{lem:cover} with $\delta=1/4$, and Lemma~\ref{lem:haarball} with $R=1/2$, $a=3/4$. Lemma~\ref{lem:separation} then gives \eqref{eq:lower} with \eqref{eq:Bn}--\eqref{eq:eta}.
\end{proof}

\begin{proof}[Proof of Theorem~\ref{thm:stronglower}]
Fix any $0<r<\sqrt2$, and choose $\delta>0$ such that $R=r+\delta<\sqrt2$. Choose any $0<a<1-R^2/2$. By \eqref{eq:entropyasymptotic}, the support has a $\delta$-net with logarithmic cardinality $o(n^2)$. Lemma~\ref{lem:haarball} bounds every $R$-ball by $\exp\{-c(R,a)n^2+O_{R,a}(n)\}$. Their product tends to zero, uniformly over every deterministic choice. Lemma~\ref{lem:separation} therefore implies
\[
 \liminf_{n\to\infty}W_{1,\rho}(\nu_{n,m_n}^{A,x},H_n)\geq r.
\]
Since $r<\sqrt2$ is arbitrary, this liminf is at least $\sqrt2$. On the other hand, $W_{1,\rho}\leq W_{2,\rho}\leq\sqrt2$, the last inequality being \eqref{eq:independentbound}. Both distances must converge to $\sqrt2$.
\end{proof}

The convergence statement by itself is formulated for $m=o(n/\log n)$. To obtain the quantitative $\Omega(n/\log n)$ mixing-time bound, we retain constants in the entropy comparison rather than merely negating that little-oh condition.

\begin{proof}[Proof of Corollary~\ref{cor:mixing}]
The upper bound is Theorem~\ref{thm:upper}. For the lower bound fix $\eps<\sqrt2$, choose $r$ with $\eps<r<\sqrt2$, and then choose $\delta>0$ with $R=r+\delta<\sqrt2$. Fix $a\in(0,1-R^2/2)$ and write $c_0=c(R,a)>0$.
For $1\leq m\leq n$, Lemma~\ref{lem:cover} gives
\begin{align*}
 \log N_\rho(S_{n,m}^{A,x},\delta)
 &\leq mn\log n+mn\log\bigl(1+\pi\sqrt2\,n/\delta\bigr)\\
 &\leq mn\bigl(2\log n+C_\delta\bigr)
\end{align*}
for a finite constant $C_\delta$. Thus if
$m\leq(c_0/4)n/\log n$, the logarithm of the product of this covering cardinality and the Haar ball bound is at most
\[
 -c_0n^2+O(n)+\frac{c_0n^2}{4\log n}
 (2\log n+C_\delta)
 =-\frac{c_0}{2}n^2+o(n^2).
\]
Uniformly throughout this range, Lemma~\ref{lem:separation} gives
$W_{1,\rho}(\nu_{n,m}^{A,x},H_n)\geq r(1-o(1))>\eps$ for all sufficiently large $n$. Since $W_2\geq W_1$, such $m$ cannot suffice. At $m=0$, $W_{2,\rho}(\delta_{A_0x},H_n)=\sqrt2$ by direct integration. Reducing the constant slightly absorbs integer endpoints and proves \eqref{eq:mixinggap} with, for example, any fixed $c_\eps<c_0/4$ once $n_\eps$ is large enough. The estimates did not depend on the insertion matrices or their locations, so optimization over those choices cannot evade the lower bound.
\end{proof}

\subsection{Fixed absolute accuracy}\label{subsec:absoluteproof}

The absolute-accuracy lower bound uses the same support-versus-Haar comparison as before, but now the relevant balls shrink like $n^{-1/2}$ in the normalized metric. At that scale, the logarithm of the support covering number is of order $mn\log n$, whereas a fixed-radius ball in the unnormalized Frobenius metric has Haar mass of order $\exp\{-c n^2\log n\}$. Balancing these exponents rules out mixing for $m$ of order $n$ with a sufficiently small constant. The proof below keeps explicit constants only to obtain the convenient value $1/100$; no optimization is intended.

\begin{proof}[Proof of Corollary~\ref{cor:absolutemixing}]
Fix $\eps>0$ and take $n$ sufficiently large, in particular $n>18\eps^2$. The upper bound follows from \eqref{eq:unnormalized}. We prove the lower bound uniformly over all deterministic factors, initial states, and integers $1\leq m\leq n/100$.

First cover $S=S_{n,m}^{A,x}$ by balls of radius $\eps$ in $\dF$, equivalently radius $\delta_n=\eps/\sqrt n$ in $\rho$. Lemma~\ref{lem:cover} gives a net with cardinality at most
\begin{equation}\label{eq:absoluteentropy}
 \mathcal L_{n,m,\eps}
 =(n!)^m\left\lceil\frac{\pi\sqrt2\,m\sqrt n}{\eps}\right\rceil^{m(n-1)}.
\end{equation}
For $1\leq m\leq n$, use $n!\leq n^n$ and
\[
 \left\lceil\frac{\pi\sqrt2\,m\sqrt n}{\eps}\right\rceil
 \leq1+\frac{\pi\sqrt2}{\eps}n^{3/2}
 \leq\left(1+\frac{\pi\sqrt2}{\eps}\right)n^{3/2}.
\]
With $C_\eps=\log(1+\pi\sqrt2/\eps)$ this yields the uniform bound
\begin{equation}\label{eq:absolutelogentropy}
 \log\mathcal L_{n,m,\eps}
 \leq\frac52 mn\log n+C_\eps mn.
\end{equation}

Next bound every Haar ball of radius $3\eps$ in $\dF$. In the nonasymptotic formula \eqref{eq:haarparameters}, choose
\[
 R_n=\frac{3\eps}{\sqrt n},\qquad
 a_n=1-\frac{9\eps^2}{n},\qquad
 t_n=1-\frac{9\eps^2}{2n}.
\]
The condition $n>18\eps^2$ ensures $0<R_n<\sqrt2$ and $0<a_n<t_n<1$. Furthermore
\[
 \frac{t_n-a_n}{1-a_n}=\frac12,
 \qquad k=\lfloor n/2\rfloor,
 \qquad 1-a_n^2\leq\frac{18\eps^2}{n}<1.
\]
Consequently Lemma~\ref{lem:haarball} gives, uniformly in $V\in\OO(n)$,
\begin{equation}\label{eq:absoluteball}
 H_n\{U:\dF(U,V)\leq3\eps\}
 \leq\mathcal B_{n,\eps}
 :=\min\left\{1,\,2^{n-k}
       \left(\frac{18\eps^2}{n}\right)^{k(2n-k-1)/4}\right\}.
\end{equation}
We are using the lemma's finite-$n$ formula, not \eqref{eq:ballasymptotic}, which was stated for fixed $R,a$. To check the leading term explicitly, write $k=n/2-\eta_n$ with $0\leq\eta_n\leq1/2$. Then
\[
 \frac{k(2n-k-1)}4
 =\frac{3n^2}{16}-\frac n8
  -\frac{n\eta_n}{4}+\frac{\eta_n-\eta_n^2}{4}
 =\frac{3n^2}{16}+O(n).
\]
Taking logarithms of the untruncated expression in \eqref{eq:absoluteball} therefore gives
\begin{equation}\label{eq:absoluteballrate}
 \log\mathcal B_{n,\eps}
 \leq-\frac3{16}n^2\log n+O_\eps(n^2).
\end{equation}
Indeed the other terms are $(n-k)\log2$, an $O_\eps(n^2)$ term from $\log(18\eps^2)$, and an $O(n\log n)$ term from the floor error, which is also $O(n^2)$. The leading expression is negative for all sufficiently large $n$.

For $m\leq n/100$, \eqref{eq:absolutelogentropy} and \eqref{eq:absoluteballrate} imply
\begin{equation}\label{eq:absoluteseparationrate}
 \log(\mathcal L_{n,m,\eps}\mathcal B_{n,\eps})
 \leq-\left(\frac3{16}-\frac1{40}\right)n^2\log n
       +O_\eps(n^2)
 =-\frac{13}{80}n^2\log n+O_\eps(n^2).
\end{equation}
The constants are independent of $m,A,x$ in the stated range, so the product tends to zero uniformly. An $\eps$-net for $S$ and balls of radius $3\eps$ cover its closed $2\eps$-neighborhood. Apply Lemma~\ref{lem:separation} in the metric $\dF$, with $\delta=\eps$, $R=3\eps$, and $p=1$, to obtain
\[
 W_{1,\dF}(\nu_{n,m}^{A,x},H_n)
 \geq2\eps\bigl(1-\min\{1,\mathcal L_{n,m,\eps}\mathcal B_{n,\eps}\}\bigr)
 =2\eps(1-o_\eps(1))>\eps
\]
for all sufficiently large $n$. Since $W_2\geq W_1$, every integer $1\leq m\leq n/100$ is excluded. At $m=0$, \eqref{eq:independentbound} gives $W_{2,\dF}(\delta_{A_0x},H_n)=\sqrt{2n}>\eps$. Hence $t_{n,\mathrm F}^{A}(\eps)>\lfloor n/100\rfloor$, which implies \eqref{eq:absolutemixinggap} with $c_*=1/100$. Choosing $n_\eps$ large enough to satisfy all preceding conditions proves the assertion, uniformly over the deterministic choices.
\end{proof}

\section{The dimension obstruction in total variation}\label{sec:singular}

Total variation behaves differently from Wasserstein distance. For a fixed sequence of permutations, the output after $m$ sweeps is a smooth function of only $m(n-1)$ real angle parameters. As long as this number is smaller than the dimension $n(n-1)/2$ of one determinant component of $\OO(n)$, the image has Haar measure zero. This gives an elementary singularity obstruction before any quantitative metric-entropy estimate is needed.

\begin{lemma}[Smooth images of lower-dimensional parameter spaces]\label{lem:nullimage}
Let $M$ be a compact smooth $N$-dimensional Riemannian manifold and let $f:\T^d\to M$ be smooth, with $d<N$. Then $f(\T^d)$ has Riemannian volume zero.
\end{lemma}
\begin{proof}
Smoothness and compactness bound the derivative of $f$, so $f$ is Lipschitz for the intrinsic torus metric and the Riemannian distance. A grid covers $\T^d$ by $O(\eta^{-d})$ balls of radius $\eta$. Their images lie in $O(\eta^{-d})$ balls of radius $C\eta$ in $M$. A finite collection of smooth coordinate charts with bounded metric densities implies that all sufficiently small balls of radius $u$ in $M$ have volume at most $C_Mu^N$. The volume of the image is therefore at most $C'\eta^{N-d}$, which tends to zero. The case $d=0$ is the same assertion for finitely many points.
\end{proof}

\begin{proof}[Proof of Theorem~\ref{thm:TV}]
Fix the $m$ permutations. The output is a smooth function of $d=m(n-1)$ angles and has a fixed determinant on that branch. Each determinant component has dimension
$N=\dim\so(n)=n(n-1)/2$. Its normalized Riemannian volume is conditional Haar measure, because the metric is invariant under translations in $\SO(n)$; translating by one reflection identifies the other component. When $m<n/2$, we have $d<N$, so Lemma~\ref{lem:nullimage} makes every branch Haar-null. There are only $(n!)^m$ branches, and their finite union remains Haar-null and compact. The output assigns that union probability one, whereas Haar measure assigns it zero. The total-variation distance is consequently one.
\end{proof}

\begin{remark}
The converse is not asserted: $m\geq n/2$ does not, by dimension counting alone, imply absolute continuity, a density lower bound, or total-variation mixing. Nor does singularity alone imply a Wasserstein lower bound. The separate neighborhood argument in Sections~\ref{sec:smallball}--\ref{sec:lowerproof} is needed because a low-dimensional set can in principle be close to a large fraction of a higher-dimensional space.
\end{remark}

\section{Two extensions}\label{sec:extensions}

\subsection{A determinant-corrected sampler on $\SO(n)$}

Define
\begin{equation}\label{eq:correction}
 C(U)=J_{\det U}U,\qquad
 J_+=I,\quad J_-=\diag(-1,1,\ldots,1).
\end{equation}
Then $C(U)\in\SO(n)$, and $C_\#H_n=H_{n,+}$.

\begin{proposition}\label{prop:SO}
The upper bound of Theorem~\ref{thm:upper}, the explicit lower bound of Theorem~\ref{thm:lower}, the limits of Theorem~\ref{thm:stronglower}, and the singularity assertion of Theorem~\ref{thm:TV} remain valid for $C(X_m)$ with target $H_{n,+}$.
\end{proposition}
\begin{proof}
On a fixed determinant component, $C$ is left multiplication by a fixed orthogonal matrix. In every determinant-matched coupling used for the upper bound, the two matrices receive the same multiplier, so their $\rho$ distance is unchanged. Pushing those couplings forward gives the same $b_nr_n^m$ upper bound on $\SO(n)$.

For the lower bounds, each fixed permutation branch has a constant determinant. Therefore $C$ is an isometry on that branch and its angle grid, and the support covering numbers in Lemma~\ref{lem:cover} are unchanged as upper bounds. Lemma~\ref{lem:haarball} explicitly applies to conditional Haar measure, so the neighborhood argument and its asymptotic version apply without alteration. The upper bound by $\sqrt2$ also holds on $\SO(n)$ for $n\geq2$: every matrix entry has conditional Haar mean zero, since left multiplication by a diagonal matrix reversing its row and one other row preserves determinant and changes its sign. Independence then gives \eqref{eq:independentbound} with $H_{n,+}$.

Finally, $C$ is smooth on each branch, so the image-dimension proof gives the same singularity threshold. This proof uses componentwise isometry; it does not assume that $C$ is globally nonexpanding between opposite determinant components.
\end{proof}

\subsection{Ordered rooted-tree sweeps}

Let $p(j)\in\{1,\ldots,j-1\}$ for $2\leq j\leq n$. These parent choices specify a rooted tree whose vertices are ordered so that a parent precedes each child. Replace \eqref{eq:sweep} by
\begin{equation}\label{eq:tree}
 Q_{\mathcal T}(\theta)
 =R_{p(n),n}(\theta_{n-1})\cdots R_{p(2),2}(\theta_1).
\end{equation}
Each rotation introduces one previously untouched coordinate. The path sweep is the choice $p(j)=j-1$; the star sweep is $p(j)=1$.

\begin{proposition}\label{prop:tree}
All four main theorems remain valid for \eqref{eq:tree} followed on the right by an independent uniform permutation. The tree may be chosen deterministically and may differ from one sweep to the next. The determinant-corrected statements also remain valid.
\end{proposition}
\begin{proof}
Let $S_{j-2}$ be the product of rotations preceding the rotation that introduces vertex $j$. It acts only on coordinates $1,\ldots,j-1$, and the corresponding angle derivative is
\[
 B_{j-1}=u_j e_j^{\mathsf T}-e_j u_j^{\mathsf T},\qquad
 u_j=S_{j-2}^{-1}e_{p(j)}.
\]
The upper-triangular support of this derivative lies in column $j$, and $\|u_j\|=1$. Thus the derivatives are orthonormal and the differential identity \eqref{eq:flatdifferential} still holds. They depend only on preceding angles, as needed for the triangular coupling.

It remains to check diagonal covariance. For any diagonal matrix $D_0$ and any coordinate-plane rotation $R(\theta)$ with uniform angle, the matrix
$\E[R(\theta)^{-1}D_0R(\theta)]$ is diagonal: its two affected diagonal entries are replaced by their average and all off-diagonal means vanish. Induct on the length of a product of independent such rotations. If $S_\ell=R_\ell S_{\ell-1}$, first average
\[
 S_\ell^{-1}D_0S_\ell
 =S_{\ell-1}^{-1}R_\ell^{-1}D_0R_\ell S_{\ell-1}
\]
over the last angle, producing a deterministic diagonal matrix inside the two prefix factors. The inductive hypothesis then makes the full expectation diagonal. Apply this to $D_0=e_{p(j)}e_{p(j)}^{\mathsf T}$ and $S_{j-2}$ to conclude that $\E[u_ju_j^{\mathsf T}]$ is diagonal, with trace one.

The permutation averaging argument in Lemma~\ref{lem:isotropy} is now unchanged. Consequently the same exact triangular coupling, infinitesimal contraction, and componentwise globalization prove the upper bound. There are still $n-1$ angle parameters per sweep, the same differential norm, and the same number of permutation branches. Hence the covering and dimension arguments give all lower bounds. Allowing a different deterministic tree at each step changes none of these estimates. Determinant correction is covered by the same componentwise argument as in Proposition~\ref{prop:SO}.
\end{proof}

Both mixing-time corollaries also hold for these extensions. The determinant correction is an isometry on each permutation branch, the ordered-tree sweeps have the same parameter count and differential bound, and the small-ball bound holds for conditional Haar measure. Thus the covering and ball estimates used in the proof of Corollary~\ref{cor:absolutemixing} remain valid. The upper bounds follow from Propositions~\ref{prop:SO}--\ref{prop:tree}, and the deterministic-start absolute $W_2$ distance on $\SO(n)$ is again $\sqrt{2n}$ for $n\geq2$.

\section{Interpretation and remaining quantitative questions}
\label{sec:discussion}

\subsection{Full Haar approximation versus statistical diagnostics}

The sampler was previously investigated using eigenvalue-based statistics and
a statistic for the full group \cite{Sepehri19}. In the reported
$51$-dimensional experiment, the full-group statistic detected departure from
Haar measure after one and two sweeps, but not at the larger sweep counts
considered \cite[Section~4 and Table~7]{Sepehri19}. The corresponding test is
consistent against every fixed alternative as the sample size tends to
infinity \cite[Remark~4.7]{Sepehri19}. Such consistency, however, does not
provide a dimension-uniform quantitative bound on the distance between a
nonrejected distribution and Haar measure.

The present results illustrate this distinction. For $n=51$,
Theorem~\ref{thm:TV} shows that the output law is singular with respect to
Haar measure through $25$ sweeps. More strongly,
Theorem~\ref{thm:stronglower} shows that if $m=o(n/\log n)$, then
\[
 W_{1,\rho}(\nu_{n,m}^{A,x},H_n),
 \quad
 W_{2,\rho}(\nu_{n,m}^{A,x},H_n)
 \longrightarrow \sqrt2 .
\]
The value $\sqrt2$ is extremal in this normalization: every probability
measure $\mu$ on $\OO(n)$ satisfies
$W_{2,\rho}(\mu,H_n)\leq\sqrt2$, while
$W_{2,\rho}(\delta_V,H_n)=\sqrt2$ for every deterministic
$V\in\OO(n)$. Thus, on this time scale, the full matrix law is
asymptotically no closer to Haar in normalized Wasserstein distance than a
deterministic orthogonal matrix. This does not contradict the finite-sample
diagnostic results, since particular observables may approach their Haar
distributions substantially earlier than the complete matrix law.

The bounds are also uniform over predetermined orthogonal insertions:
such factors are isometries for the upper-bound coupling and do not alter
the parameter count or permutation branching used in the lower bounds.
Thus the Fourier factor does not affect the quantitative estimates proved
here, although it may still influence particular observables or
application-specific performance.

\subsection{The logarithmic gap}

The principal quantitative question left open by the present analysis is the
remaining logarithmic gap between the lower and upper bounds. At fixed
normalized Frobenius accuracy,
\[
 \Omega\left(\frac{n}{\log n}\right)
 \leq
 t_{\mathrm{mix}}
 \leq
 O(n),
\]
whereas at fixed absolute Frobenius accuracy,
\[
 \Omega(n)
 \leq
 t_{\mathrm{mix}}
 \leq
 O(n\log n).
\]

The two sides arise from rather different mechanisms. The upper bound is a
local contraction estimate: one sweep removes, on average, a fraction
$2/n$ of an infinitesimal squared displacement, leading to exponential
contraction on the scale of $n$ sweeps. The lower bound is geometric. It
compares the metric entropy of the sampler's support with the Haar measure
of small metric balls. At fixed normalized radius, the support entropy
becomes negligible relative to the Haar small-ball exponent only when
$m=o(n/\log n)$, yielding the lower scale $n/\log n$.

Closing the gap therefore appears to require information not captured by
either argument alone. A sharper lower bound would need to exploit more than
the dimension and covering entropy of the set of possible outputs, for
example the geometry or concentration of the induced measure on that set.
Conversely, an improved upper bound would require substantially more than
the first-order contraction estimate proved here, since that estimate gives
the natural time scale $n$.

It remains open whether the correct fixed-normalized-accuracy mixing time is
of order $n/\log n$, of order $n$, or lies strictly between these scales.
Equivalently, at fixed absolute Frobenius accuracy it is open whether the
correct order is $n$, $n\log n$, or intermediate. Determining the sharp
order, and whether either formulation exhibits cutoff, are the main
quantitative problems left by this work.

\begin{acks}[Acknowledgment and disclosure note]
The author is grateful to Persi Diaconis for feedback on an early version of the manuscript that helped improve the writing. AI-assisted tools were used in the preparation of this manuscript for literature search, computational checking, exploring proof ideas, and editorial assistance, as well as auto-formalization using Lean. The author takes full responsibility for the mathematical arguments, accuracy, and final text.
\end{acks}

\bibliographystyle{imsart-nameyear}
\bibliography{jor_mixing_references}
\end{document}